\documentclass[sigconf]{acmart}

\usepackage{graphicx}
\usepackage{caption}
\usepackage{subcaption}
\usepackage{amsmath}
\usepackage{mleftright}

\theoremstyle{remark}
\newtheorem*{remark}{Remark}

\AtBeginDocument{%
  \providecommand\BibTeX{{%
    \normalfont B\kern-0.5em{\scshape i\kern-0.25em b}\kern-0.8em\TeX}}}

\setcopyright{acmcopyright}
\copyrightyear{2022}
\acmYear{2022}
\acmDOI{XXXXXXX.XXXXXXX}

\acmConference[Preprint '22]{Make sure to enter the correct
  conference title from your rights confirmation emai}{July 14--15, 2022}{Madrid, Spain}
\acmPrice{15.00}
\acmISBN{978-1-4503-XXXX-X/18/06}

\begin{document}

\title{Noise tolerance of learning to rank under class-conditional label noise}

\author{Dany Haddad}
\email{danyhaddad@utexas.edu}

\renewcommand{\shortauthors}{Haddad}

\begin{abstract}
Often, the data used to train ranking models is subject to label
noise. For example, in web-search, labels created from clickstream
data are noisy due to issues such as insufficient information in
item descriptions on the SERP, query reformulation by the user, and
erratic or unexpected user behavior. In practice, it is difficult to
handle label noise without making strong assumptions about the label
generation process. As a result, practitioners typically train their
learning-to-rank (LtR) models directly on this noisy data without
additional consideration of the label noise. Surprisingly, we often
see strong performance from LtR models trained in this way. In this
work, we describe a class of noise-tolerant LtR losses for which
empirical risk minimization is a consistent procedure, even in the
context of class-conditional label noise. We also develop
noise-tolerant analogs of commonly used loss functions. The
practical implications of our theoretical findings are further
supported by experimental results.
\end{abstract}



\maketitle

\section{Introduction}   
Practical learning-to-rank models are often trained on datasets with incomplete or noisy labels. A dataset might be subject to label noise due to various reasons, such as adjudication errors by manual reviewers, incomplete identification of relevant documents, and unexpected user behavior in clickstream data. Historically, practitioners have trained models directly on data with noisy labels without explicit justification for the process.

In this work, we explore the noise tolerance properties of commonly used loss functions in learning to rank and develop theoretical results which explain the success of certain LtR models trained on noisy data.
Specifically, we describe a class of noise-tolerant losses, which we refer to as \textit{order-preserving}, for which
empirical risk minimization (ERM) is a consistent procedure even in the
presence of class-conditional label noise. Further, we show that
classical statistical learning theory results go through, almost without
modification.
We also identify a sufficient condition for a loss function to be considered order-preserving. While not all commonly used loss functions satisfy this condition, we introduce order-preserving analogs of these losses. The applicability of these losses to practical applications is explored in the experimental results section.
The results developed here provide further theoretical justification
for the use of weak supervision methods in learning to rank \cite{dehghani2017,zamani2018acm,haddad2019}.

Our results also have implications for training classification
models. In practice, we can train a ranking
model using the available noisy data, and separately collect a small subset of
clean data for determining the classification threshold. In critical applications, we can use ideas
from distribution-free risk control to identify a cutoff point that
guarantees a limited level of risk, with high probability \cite{angelopoulos2021,bates2021}.
    
\section{Related work}
It is well known that optimal classifiers learned in the presence of
class-conditional noise differ from the optimum under the clean
distribution only in terms of the classification threshold. Natarajan
et al. build their method of label dependent costs based on this
observation
\cite{natarajan2018}. Zhang et al. uses a similar observation to
design their post-processing based approach \cite{zhang2021}. As they
are primarily interested in learning a classifier, neither
of these works considers the ranking problem, or the implications of
this observation on empirical risk minimization of ranking
objectives. In this work, we consider the same class-conditional label noise condition as \cite{natarajan2018,zhang2021} and study it's implications for learning to rank.

It is immediate from the above observation on the role of the
classification threshold, that the
optimum for noisy class probability estimation is optimal for both the
noisy ranking problem and the clean ranking problem. This fact
suggests that minimization of ranking objectives has certain noise tolerance properties. In our analysis, we do
not necessarily assume that the optimum of the noisy risk is obtained. Instead, we are interested in the more practical case where we may face optimization
difficulties over our model class or be limited by the size of the training set
available.

Zamani and Croft study the behavior of so-called symmetric loss functions under noise conditions which they refer
to as uniform and non-uniform noise \cite{zamani2018ictir}. They find that these symmetric losses are noise tolerant in the sense that the global optimum of the noisy empirical risk is also a global optimum of the clean empirical risk. As it
turns out, symmetry of the loss function alone is not sufficient to
ensure their noise tolerance condition. Indeed, they
implicitly require that the corrupted labels be selected uniformly at
random over the set of incorrect labels (see their equation 17 and the
first equality in equation 21). This assumption is particularly
strong, especially when treating the set of labels as a single label
over $\{0, 1\}^n$ as is the case in their setup. Accordingly, although
NDCG and other learning to rank metrics satisfy the symmetry
condition, their theorems 2 and 3 do not apply in the general class
conditional noise regime which we consider in this work. Further,
their results are concerned only with the global optima, while our
results are more general and have strong implications for empirical
risk minimization with finite samples. Finally, it is worth noting
that the proofs of their theorems 2 and 3 deal with the
\textit{expected} empirical risk (the expectation taken over the noise), rather than the empirical risk which
we can compute in practice.

The notion of noise tolerance we explore in this work (which we refer
to as order preservation) is reminiscent of the property of
\textit{consistent distinguishability} explored by Wang et al. in
\cite{wang13}. However, their analysis is purely concerned with the
behavior of $\text{NDCG@k}$ losses as $k \rightarrow \infty$, with no
consideration of the effects of label noise.
    
\section{Background and Notation}
  Define $\mathcal X, \mathcal Y, \mathcal Q$ as the feature space, label space and query space, respectively. Let $\mathbb P_q$ be a distribution over $\mathcal X \times \mathcal Y$ associated with a query, $q \in \mathcal Q$, from which we have access to noisy iid samples $(X_i^q, \tilde Y_i^q)$ where $\tilde Y_i^q$ is a corrupted version of $Y_i^q$. Assume further that there exists a distribution over queries, so that query $q$ is issued with probability $\nu_q$, or formally $\mathbb P (Q = q) = \nu _q$ where $Q$ is the random variable associated with the issued query. We drop references to the query $q$ where the additional specificity is not needed.

Similar to previous work, we restrict ourselves to the case of binary relevance, $\mathcal Y = \{0, 1\}$﻿. From a practical information retrieval perspective, we can consider $i$ as indexing the instances, so that $X_i$ represents the features of document $i$ and $Y_i$ its relevance label. To accommodate pairwise losses, we define $Y_{ij} = (Y_i - Y_j + 1)/2$ so that ties are indicated by $1/2$.

In this work, we assume that our dataset has been corrupted with
\textit{class-conditional label noise}: corruptions to the labels that
are independent of the features, conditioned on the true
class. Further, for simplicity of exposition, without loss of generality,
the derivations assume that the probability of corruptions are the
same for each class, $\mathbb P (\tilde Y = 1 \mid Y = 0) = \mathbb P
(\tilde Y = 0 \mid Y = 1)$.\footnote{Note that the results in this work go
through almost without modification if we do not make this assumption.} Putting this together, the corrupted label
is given by:

\begin{align*}
  \tilde Y_i = \varepsilon_i Y_i + (1-\varepsilon_i) (1-Y_i),
\end{align*}

 \noindent where the corruptions are determined by $\varepsilon_i$ which
 are iid Bernoulli random variables with parameter $\gamma \in [0,
 1]$ and are independent of the features, conditioned on the true class:

 \begin{align*}
   \varepsilon_i \perp\!\!\!\!\!\perp X_i \mid Y_i.
 \end{align*}

 Note that
 when $\varepsilon_i = 1$ we receive the correct value of $Y$, so larger
 values of $\gamma$ imply less noise.

 In our setup, we assume that the noise is constant across queries. In
 the case where the label noise can vary arbitrarily for query to query, we don’t have
 any guarantees (without other additional assumptions) since the noise
 can be chosen such that the expected noisy risk incorrectly prefers one scoring
 function to another.

 Let $\mathcal F = \{ f \mid f : \mathcal X \rightarrow \mathbb R\}$ be the class of scoring functions over which we are optimizing the risk $L^\ell: \mathcal F \rightarrow \mathbb R$ for a given margin based loss $\ell : \mathbb R  \rightarrow \mathbb R^+$. If $\ell$ is a pointwise loss, then
 $$L^\ell(f) = \mathbb E [\ell (f(X)(2Y-1))].$$
 If $\ell$ is a pairwise loss then
 $$L^\ell(f) = \mathbb E [\ell ((f(X_1) - f(X_2))(2Y_{12}-1)) \mid Y_{12} \neq 1/2].$$  We use $\tilde L^\ell$ to indicate the noisy risk, so for pointwise losses:
 $$\tilde L^\ell(f) = \mathbb E [\ell(f(X) (2\tilde Y-1))],$$
 and for pairwise losses:
 $$\tilde L^\ell(f) = \mathbb E [\ell((f(X_1) - f(X_2))(2\tilde Y_{12}-1)) \mid \tilde Y_{12} \neq 1/2].$$ Notice that in both the clean and the noisy case, the pairwise risk is computed only over pairs of documents with different labels. When the choice of loss function is clear from context, we use $L$ to indicate the risk $L^\ell$.
We indicate the empirical versions of the risk with the subscript $n$; the clean and noisy risk, respectively:
$$L^\ell_n(f) = \frac{1}{n}\sum_i^n \ell(f(X_i)(2Y_i-1)),$$
$$\tilde L^\ell_n(f) = \frac{1}{n}\sum_i^n \ell(f(X_i)(2\tilde Y_i-1)).$$

The key results of this work are concerned with the class of loss functions with the following noise
tolerance property.

\begin{definition}[order-preserving loss functions]
  A loss function $\ell$ is \textit{order-preserving} if the following condition holds
  for all $f, g \in \mathcal F$ and distributions over $\mathcal X
  \times \mathcal Y$ for which $h(X)$ has a density with respect to Lebesgue measure for every $h \in \mathcal F$:
  \begin{align}\label{eq:order-pres}
    \tilde L^\ell(f) - \tilde L^\ell(g) = (L^\ell(f) - L^\ell(g))(2\gamma - 1).
  \end{align}
\end{definition}

The order preserving property implies that when the noise level is not
too high, then scoring functions that are preferred by the clean risk
are also preferred by the noisy risk and vice versa. This is a
powerful property that (as we will see in section \ref{sec:conv})
allows us to derive finite-sample guarantees on the convergence of empirical risk
minimization of order-preserving loss functions. Note that the order-preserving
property is distinct from (but reminiscent of) the notion of noise tolerance considered by
Zamani and Croft \cite{zamani2018ictir}.

\begin{remark}
  In the definition of order-preserving losses, we require that $h(X)$ has a density for each $h \in \mathcal F$
  to avoid the possibility of ties in the predictions. Note that this
  requirement is satisfied whenever $X$ itself has a density with respect to
  Lebesgue measure and $h$ is continuous. Alternatively, $h(X)$ can
  always be made to be absolutely continuous with respect to Lebesgue
  measure by convolving it's distribution with that of a smooth
  approximation of a delta function.
\end{remark}

For convenience, we deal with the loss equivalents of metrics such as AUC, DCG and NDCG. Specifically, define $\text{DCG@k} : \mathcal F \rightarrow \mathbb R$ as:

$$
\text{DCG@k}(f; q) = -\sum_{i}^k \frac{\,\,\,\,\,\,Y^q_{\pi^{-1}_f(i)}}{D_i},
$$

\noindent where $\pi_f : [n] \rightarrow [n]$ is the permutation of the
documents determined by the scoring function $f$ and $D_i$ is the
discount applied to rank $i$; a typical choice is $\log(1+i)$. We define $\text{NDCG@k}$ accordingly:

$$
\text{NDCG@k}(f;q) = -\frac{\text{DCG@k}(f;q)}{\text{maxDCG@k}(q)},
$$

\noindent where $\text{maxDCG@k}(q)$ is the $\text{DCG@k}$ obtained by the
arranging the items in decreasing order of their true relevance. This is
the best achievable $\text{DCG@k}$ in query, $q$. The loss
equivalents of the remaining metrics are defined by similarly taking
their negation.

Next we define a subset of the order-preserving losses which we show contains several commonly used loss functions.

\begin{definition}[label-symmetric loss functions]
A loss function $\ell :
\mathbb R \rightarrow \mathbb R^+$ is label-symmetric if for all $\alpha
\neq 0$ we have that:
$$\ell(\alpha) + \ell(-\alpha) = c,$$
for a constant $c$.
\end{definition}

This class of losses is identical
to the class of symmetric losses introduced in \cite{zamani2018ictir}, but restricted to the case of binary relevance.
    
\section{Order preservation}\label{sec:order}

In this section, we identify a sufficient condition for a loss to be
considered order-preserving and further show that there are losses
which do not satisfy this condition, but are still
order-preserving. Building on this, we present order-preserving
versions of commonly used losses that are not themselves order-preserving. Finally, we
show that the order preservation property gives us a mechanism to
apply results from statistical learning theory to the class conditional noise regime.

\subsection{Label-symmetric losses}\label{sec:label-sym}

First, we show that both pointwise and pairwise label-symmetric losses
are order-preserving. The key observation is that label-symmetry
allows us to relate the risk due to incorrectly labeled instances to
the risk due to correctly labeled instances. The proofs for the two
types of losses are similar, with the main difference for pairwise
losses being in handling the loss terms incurred by instances of the
same class, but with different noisy labels.

\begin{theorem}
A loss function $\ell : \mathbb R \rightarrow
\mathbb R^+$ is order preserving if it is label-symmetric.
\end{theorem}

\begin{proof}
    We show this result for pointwise losses and pairwise losses separately.
    
    First, consider the pointwise case. Let $f, g \in \mathcal F$ and consider the noisy risk:
    
    \begin{align*}
      \tilde L(f) &= \mathbb E[\ell(f(X)(2\tilde Y-1))]\\
                  &= \gamma \mathbb E [\ell(f(X)(2Y-1))] + (1-\gamma) \mathbb E [\ell(-f(X)(2Y-1))],
    \end{align*}
    
    \noindent where the second equality follows from the tower property of expectation and the class-conditional noise assumption. Now, by the definition of a label-symmetric loss:
    
    \begin{align*}
      &= (2\gamma - 1) \mathbb E [\ell(f(X)(2Y-1))] + c(1-\gamma) \\
      &= (2 \gamma - 1) L(f) + c(1-\gamma),
    \end{align*}
    
    \noindent which is simply an affine transform of the clean risk. The difference in their noisy risk:
    
    $$
    \tilde L(f) - \tilde L(g) = (2\gamma -1)(L(f) - L(g)).
    $$
    
    So $\ell$ is order-preserving.
    
    Consider now the pairwise risk.
    
    \begin{align*}
    \tilde L(f) &= \mathbb E[\ell((f(X_1) - f(X_2))(2\tilde Y_{12}-1)) \mid  \tilde Y_1 > \tilde Y_2]\\
                  &= \mathbb E[\ell(f(X_1) - f(X_2)) \mid \tilde Y_1 =1, \tilde Y_2=0].
    \end{align*}

    Let us denote by $L_{y_1, y_2}(f)$ the quantity:
    \begin{align}\label{eq:breakdown}
      &\mathbb E [\ell(f(X_1) - f(X_2)) \mid \tilde Y_1 = 1, \tilde Y_2 = 0, Y_1 =y_1, Y_2 = y_2] \\
      &=\, \mathbb E [\ell(f(X_1) - f(X_2)) \mid Y_1 =y_1, Y_2 = y_2],
    \end{align}
    \noindent where the equality comes from the class-conditional noise assumption. Then expand the noisy risk as:
    
    $$
    \tilde L(f) = \sum_{(y_1, y_2) \in \{0, 1\}^2} \mathbb P( Y_1 = y_1, Y_2 = y_2 \mid \tilde Y_1 = 1, \tilde Y_2 = 0) L_{y_1, y_2}(f),
    $$
    
    \noindent where we’ve used that the noise magnitude is constant over queries. When $\ell$ is label-symmetric, we have that:
    
    \begin{align*}
      &L_{1, 0}(f) = L(f)\\
      &L_{0, 1}(f) = c - L(f)\\
      &L_{0, 0}(f) = c/2 \\
      &L_{1,1}(f) = c/2.
    \end{align*}
    
    The first two equalities are immediate. To see the remaining two, expand the $L_{y, y}(f)$ terms as:
    
    \begin{align*}
      &\mathbb E[\ell(f_1 - f_2) \mid Y_1 = Y_2 = y, f_1 > f_2]\, \mathbb P(f_1 > f_2 \mid Y_1 = Y_2 = y) \\
      &+\,  \mathbb E[\ell(f_1 - f_2) \mid Y_1 = Y_2 = y, f_1 < f_2]\, \mathbb P(f_1 < f_2 \mid Y_1 = Y_2 = y),
    \end{align*}
    
    \noindent where $f_i := f(X_i)$ and we apply the assumption that the event $\left \{f(X_1) = f(X_2) \right\}$ occurs with probability $0$ when both items have the same relevance label. Using this same assumption and by symmetry, we have that $\mathbb P(f_1 > f_2 \mid Y_1 = Y_2=y) =\mathbb P(f_1 < f_2 \mid Y_1 = Y_2=y) = 1/2$. Now, since $\ell$ is label-symmetric we have that the quantity in the previous display is:
    
    \begin{align*}
      &\frac{1}{2}\mathbb E[\ell(f_1 - f_2) \mid Y_1 = Y_2 = y, f_1 > f_2] \\
      &+\, \frac{1}{2} (c - \mathbb E[\ell(f_2 - f_1)  \mid Y_1 = Y_2 = y, f_1 < f_2]).
    \end{align*}
    
    But this reduces to simply $c/2$. 
    
    Putting it all together, we see that:

    \begin{align*}
      \tilde L(f) &= \gamma^2 L(f) + (1-\gamma)^2(c - L(f)) + \gamma(1-\gamma)c\\
      &= (2\gamma - 1) L(f) + c(1-\gamma).
    \end{align*}

    So $\ell$ is order preserving.
\end{proof}

It is straightforward to check that a loss function is
label-symmetric. However, recall that this condition is not required
for order-preservation.

\begin{proposition}
The 0-1 loss, hinge loss, $\ell_1$ loss and AUC are label-symmetric loss
functions.
\end{proposition}

\begin{proof}

  For the 0-1 loss:

  \begin{align*}
    \ell_{0-1}(\alpha) + \ell_{0-1}(-\alpha)
    =\, \mathbb I \left(\alpha < 0\right) + \mathbb I \left(-\alpha < 0\right),
  \end{align*}

\noindent which is equal to 1 for all $\alpha \neq 0$. The proof for AUC is identical. See
 \cite{zamani2018ictir} for the proof of the hinge loss and $\ell_1$
 loss.

\end{proof}
Although AUC and 0-1 loss are non-differentiable, their use as loss
functions is still practical when selecting from a finite collection
of scoring functions, such as during hyperparameter optimization.

Next we show examples of commonly used loss functions that are not
label-symmetric.

\begin{proposition}
  The exponential and binary cross-entropy
losses are not label-symmetric. Similarly the ranknet and lambdarank
losses \cite{burges2006,burges2010}
are not label-symmetric.
\end{proposition}

\begin{proof}

To show the binary cross entropy loss is not label-symmetric:

\begin{align*}
  \ell_{\text{BCE}}(\alpha) + \ell_{\text{BCE}}(-\alpha)
  =\, \log(1+ \exp(-\alpha)) + \log(1+ \exp(\alpha)),
\end{align*}

\noindent which is a finite constant when $\alpha$ is $1$, and goes to $\infty$
as $\alpha$ goes to $\infty$. Recall that the lambdarank loss is
simply a scaled version of the logistic loss (applied in a pairwise
setting) so the same counterexample shows that it is also not
label-symmetric. The ranknet loss is a specific kind of lambdarank loss.

Similarly, for the exponential loss:

\begin{align*}
  \ell_{\exp}(\alpha) + \ell_{\exp}(-\alpha)
  = \,\exp(-\alpha) + \exp(\alpha),
\end{align*}

\noindent which is a finite constant when $\alpha$ is $1$, and goes to $\infty$ as $\alpha$ goes
to $\infty$.
\end{proof}

Label-symmetry of a loss function is not necessary for it to be order
preserving. We can, however, show that some commonly used non-label-symmetric losses are indeed not order preserving.

\begin{proposition}\label{prop:not-ord}
The logistic, exponential and ranknet losses are not order preserving.
\end{proposition}

\begin{proof}
For general pointwise losses:

\begin{align*}
  \tilde L(f) = \mathbb E [\ell(f(X)(2\tilde Y - 1))] = \gamma L(f) + (1-\gamma)L^-(f),
\end{align*}

\noindent where $L^-(f) = \mathbb E[\ell(-f(X)(2Y-1))]$. When
$\ell(\alpha) = \log(1+\exp(-\alpha))$ we have the logistic loss. Let
$f_a(X) = a\mathbb P(Y=1 \mid X)$ for some positive $a$. Taking
$a \rightarrow \infty$ drives $L(f_a)$ to $0$ but $L^-(f_a)$ goes to
$\infty$, so the perfect scorer will have a noisy loss that is larger
than a bounded but random scorer. So the logistic loss is not order
preserving. The same argument applies to the exponential loss, given
by $\ell(\alpha) = \exp(-\alpha)$.

For pairwise losses, using the decomposition in (\ref{eq:breakdown}):

\begin{align*}
  \tilde L(f) &= \mathbb E [(f_1- f_2)(2\tilde Y_{12}-1) \mid \tilde Y_{12} \neq 1/2]\\
&= \sum_{(y_1, y_2) \in \{0, 1\}^2} \mathbb P( Y_1 = y_1, Y_2 = y_2 \mid \tilde Y_1 = 1, \tilde Y_2 = 0) L_{y_1, y_2}(f).
\end{align*}

To show that the ranknet loss is not order preserving, consider the
same scorer as before: $f_a(X) = a\mathbb P(Y=1\mid X)$. Recall the
ranknet loss $\ell(\alpha) = \log(1+ \exp(-\alpha))$. Similar to the
pointwise case, while $L_{1, 0}(f_a) \rightarrow 0$ as $a \rightarrow
\infty$,  $L_{0, 1}(f_a)$ is unbounded as $a \rightarrow \infty$ which is eventually larger than the noisy risk for a bounded but
random scoring function.
\end{proof}

As we will see in section \ref{sec:nonconvex}, it is possible to develop
non-convex but smooth and label-symmetric analogs of the logistic and
ranknet losses.

\begin{remark}
  The lambdarank loss does not fit into the framework defined here since it also takes as input the permutation induced by the scores under consideration. However, we expect the lambdarank loss to exhibit similar (lack of) order preservation behavior to the ranknet loss.
\end{remark}

\subsection{Symmetric versions of non-label-symmetric losses}\label{sec:nonconvex}

    It is possible to define symmetric but non-convex versions of the
    cross-entropy and ranknet losses. The non-convexity of a loss,
    although potentially intimidating, does not exclude its viability
    as an objective function. Recent results show that empirical risk
    minimization of non-convex losses can be practical, even with
    classical optimization algorithms such as gradient descent with a fixed step size \cite{mei2016}.

    Given certain conditions on the data distribution and function
    class, Mei et al. show that the set of critical points of the
    empirical risk converges to the set of critical points of the true risk, with local
    minima converging to local minima, saddles to saddles and the
    global minima to the global minima; see their theorem 2 \cite{mei2016}. So, in the case of order preserving losses with class-conditional label noise, the noisy empirical local minima converge to the noisy expected local minima, which are also local minima of the clean loss, and similarly for the global optima.

    While the log-loss is not an order preserving loss, we can define
    a label-symmetric analog.

    \begin{definition}[symmetrized losses]

      The \textit{symmetrized-logistic loss} is given by $\ell(\alpha) = 1-\sigma(\alpha) := 1-(1+\exp(-\alpha))^{-1}$ which defines a proper scoring rule, $1-\sigma((2Y-1)f(X))$ \cite{gneiting}. The symmetric equivalent of the ranknet loss is identical, with scoring rule: $1-\sigma((2Y_{ij}-1)(f_i-f_j))$.
    \end{definition}

    \begin{proposition}
      The symmetrized logistic and the symmetrized ranknet losses are
      label-symmetric loss functions and are thus order preserving.
    \end{proposition}

    \begin{proof}
      For both the symmetrized logistic and the symmetrized ranknet losses:

    $$
    \ell(\alpha) + \ell(-\alpha) = 1-\sigma(\alpha) + 1 - \sigma(-\alpha) = 1.
    $$
    So they are label-symmetric and hence order-preserving.
    \end{proof}

    Our experimental results show that the order preserving property
    of these losses justifies the additional complexity resulting from
    the non-convexity.

\subsection{DCG is order-preserving}

Since DCG, NDCG and other listwise ranking losses are computed over a
collection of instances (not individuals or simply pairs), they do not
fit into the framework of label-symmetric functions defined
previously. Recall that although DCG and NDCG losses satisfy the
definition proposed in \cite{zamani2018ictir}, their results do
not apply unless additional assumptions are made. Nevertheless, DCG losses still exhibit
the same order-preservation property introduced in section \ref{sec:label-sym}.

    \begin{proposition}\label{prop:dcg-ord}
      $\text{DCG@k}$ losses are order-preserving.
    \end{proposition}

    \begin{proof}

    Expanding the noisy DCG loss:

    \begin{align*}
      \widetilde{\text{DCG@k}}(f;q) &= \sum_i^k\frac{\varepsilon^q_{\pi^{-1}_f(i)} Y^q_{\pi^{-1}_f(i)} + \left(1-\varepsilon^q_{\pi^{-1}_f(i)}\right ) \left(1-Y^q_{\pi^{-1}_f(i)}\right)}{D_i} \\
      &= \sum_i^k\left(2\varepsilon^q_{\pi^{-1}_f(i)} -1\right)\frac{Y^q_{\pi^{-1}_f(i)}}{D_i} + \sum_i^k\frac{1-\varepsilon^q_{\pi^{-1}_f(i)}}{D_i}.
    \end{align*}

    Taking the expectation we have:

    $$
    \mathbb E\left[\widetilde{\text{DCG@k}}(f;q)\right] = (2\gamma - 1) \mathbb E \left [\text{DCG@k}(f;q)\right ] + \sum_i^k\frac{1-\gamma}{D_i}.
    $$

    This is simply an affine transform of $\mathbb E \left [\text{DCG@k}(f;q)\right ]$, so DCG losses are
    order preserving.
    \end{proof}

    While both AUC and DCG losses are order preserving, AUC has a further property that makes it suitable as a metric in practice. In the proof of proposition \ref{prop:dcg-ord}, we consider the expectation over the query distribution. In practice, the same query is issued many times in order to form a high quality estimate of the DCG loss, appealing to the law of large numbers. If each query is issued only a single time, estimates of the DCG could be dominated by noise. In contrast, the samples used in the estimation of AUC are exchangeable, so that estimates of the query-wise AUC are close to the true AUC given a \textit{single} issue of a query and enough labels for that query.
    
\subsection{Finite sample results}\label{sec:conv}
    
The results presented so far deal primarily with the expected risk, rather than the empirical risk. In particular, the order-preservation property is a property of the expected risk, so the relationship in (\ref{eq:order-pres}) will not necessarily hold for all finite $n$. To formulate analogous finite sample results, we can leverage concentration inequalities, with the goal of showing that order preservation holds with high probability for most pairs under consideration.
    
    However, in practice, we are typically interested in the simpler
    task of identifying a single scoring function that achieves the optimal risk. The following results show that ERM over the noisy risk enjoys the same classical convergence guarantees as the clean case. In particular, minimization of the empirical risk over a class of finite VC dimension is consistent, achieving the same convergence rate as in the clean problem, scaled by a factor related to the noise magnitude.
    
    \begin{lemma}\label{lem:scale}
      Let $\ell$ be an order-preserving loss. If $\gamma \in (0.5, 1]$ then we have the following equality for each $\epsilon > 0$:

      \begin{align}\label{eq:scale}
        \mathbb P \left( L^\ell(f_n) - \inf_{f \in \mathcal F} L^\ell(f) > \epsilon \right) = \mathbb P \left( \tilde L^\ell(f_n) - \inf_{f \in \mathcal F} \tilde L^\ell(f) > \epsilon  (2\gamma -1) \right),
      \end{align}
\noindent    where $f_n$ is the function selected by empirical risk minimization with $n$ samples.
    \end{lemma}

    \begin{proof}
    The proof is immediate from the definition of order preserving losses.
    \end{proof}

    \begin{corollary}
      Assume that for the function class $\mathcal F$ and a loss $\ell$ we have the following uniform law of large numbers:

      \begin{align*}
        \sup_{f \in \mathcal F} \left|  L_n^\ell(f) - L^\ell(f)\right|
        \overset{p}{\rightarrow} 0.
      \end{align*}

      Then running empirical risk
      minimization on the noisy estimates of the risk is equivalent
      (up to constant scaling factors of the convergence rate) to running ERM using clean estimates.
    \end{corollary}
    
    The above corollary is a direct consequence of classical results
    from statistical learning theory
    \cite{devroye1997,wainright2019} and lemma \ref{lem:scale}. Note that the deviation
    $\epsilon$ on the RHS of (\ref{eq:scale}) is simply scaled by a term related to the
    magnitude of the noise, so the convergence rate is accordingly
    scaled. For example, for the 0-1 loss, by applying Devroye et al.'s theorem 12.6
    \cite{devroye1997} to the RHS of (\ref{eq:scale}) we can
    bound the excess risk of the model selected using ERM, with high probability:

    \begin{theorem}\label{thm:dev}
      For each $\epsilon > 0$ and number of samples $n$:
\begin{align}\label{eq:dev}
      \mathbb P \left ( L(f_n) - \inf_{f\in\mathcal F}L(f) > \epsilon \right ) \leq 8 \mathcal S(\mathcal F, n) e^{-n\epsilon^2(2\gamma - 1)^2/128},
    \end{align}

\noindent    where $\mathcal S(\mathcal F, n)$ gives the $n^{th}$ shatter
    coefficient of the function class $\mathcal F$.

  \end{theorem}

  Theorem \ref{thm:dev} shows that if the function
  class is not too complex (in terms of the growth rate of it's
  shatter coefficient), then ERM is a consistent procedure. In
  particular, the risk of the selected scoring function,
  $f_n$, converges to the optimal risk achievable by any function in the class, $\mathcal F$.

     Inverting the
    deviation bound gives us a bound on the difference between the
    expected risk and the optimal risk when we only have access to
    noisy data:
    
    $$
    \mathbb E [L(f_n)] - \inf_{f \in \mathcal F} L(f) \leq 16 \sqrt{\frac{\log(8e\mathcal S(\mathcal F, n))}{2n(2\gamma - 1)^2}},
    $$
    
    \noindent where the expectation is taken over the randomness in the dataset
    used for empirical risk minimization. Notice that these bounds are
    non-asymptotic and are valid for \textit{every} value of $n$. Similar bounds can be
    derived for other order-preserving losses from their classical analogs,
    see \cite{agarwal2005} and \cite{gyorfi2002}.

    \subsubsection{Noise tolerance of almost-optima}

    We can also develop similar deviation bounds in the case where
    empirical risk minimization does not succeed. Since an empirical
    minimizer may be computationally difficult to obtain in practice,
    we now assume that our noisy empirical risk is only
    \textit{almost-minimized}, in the sense that our selected function
    is close to the optimum with high probability. More formally,
    assume that given $n$ samples, the noisy excess empirical risk is
    bounded above by $\epsilon_n$ with probability as least
    $1-\delta_n$ for some constants $\epsilon_n, \delta_n > 0$:

    \begin{align}\label{eq:approx-min}
      \mathbb P \left ( \tilde L_n(f_n) - \inf_{f \in \mathcal F}
      \tilde L_n(f) \leq \epsilon_n \right) \geq 1-\delta_n.
    \end{align}

    Then we have the following non-asymptotic result:

    \begin{theorem}\label{thm:erm-approx-min}
      Let $\ell$ be an order-preserving loss. Assume that we almost-minimize the noisy empirical risk,
      as in (\ref{eq:approx-min}). Then, for each $n, \epsilon > 0$ we have that:
      \begin{align*}
        &\mathbb P \left ( L(f_n) - \inf_{f \in \mathcal F} L(f) >
        \epsilon \right) \\
        & \leq \delta_n +
        \mathbb P \left( 2 \sup_{f \in \mathcal F} \mid \tilde L_n(f)
        - \tilde L(f) \mid > \epsilon(2\gamma -1) - \epsilon_n \right),
      \end{align*}

      \noindent where, as before, $f_n$ is the function selected by minimization
      of the noisy empirical risk.
    \end{theorem}

    \begin{proof}
      First, since $\ell$ is order preserving, we have that:
      $$L(f_n) - \inf_{f \in \mathcal F} L(f) = (\tilde L(f_n) - \inf_{f \in
        \mathcal F}\tilde L(f))/(2\gamma - 1).$$
      Now, rewrite $\tilde L(f_n)
      - \inf_{f \in \mathcal F} \tilde L(f)$ as:

      \begin{align*}
        \tilde L(f_n)-\tilde L_n(f_n) + \tilde L_n(f_n) -\inf_{f \in
        \mathcal F} \tilde L_n(f) +
        \inf_{f \in \mathcal F} \tilde L_n(f) - \inf_{f \in \mathcal F}
        \tilde L(f),
      \end{align*}

      \noindent which with probability at least $1-\delta_n$ is bounded above by:
      \begin{align*}
        \leq 2\sup_{f \in \mathcal F} \mid \tilde L_n(f) - \tilde L(f) \mid + \,\epsilon_n.
      \end{align*}

      So, by a union bound, we have the claimed inequality.
    \end{proof}

    Now, as before, applying classical results from statistical
    learning theory, we can derive deviation bounds for symmetric loss
    functions when the risk is only almost-minimized. For
    example, in the case of the 0-1 loss we have the following result:

    \begin{theorem}
      Let $\ell$ be the 0-1 loss. Assume that we almost-minimize the noisy empirical risk,
      as in (\ref{eq:approx-min}). Then, for each $n, \epsilon > 0$ we have that:
      \begin{align*}
        \mathbb P \left ( L(f_n) - \inf_{f \in \mathcal F} L(f) >
        \epsilon \right)
        \leq \delta_n + 8 \mathcal S(\mathcal F, n) e^{-n(\epsilon(2\gamma - 1)-\epsilon_n)^2/128}.
      \end{align*}
    \end{theorem}

    \begin{proof}
      The result is immediate from theorem \ref{thm:erm-approx-min} and Devroye et al. theorem 12.6 \cite{devroye1997}.
    \end{proof}
    
\section{Experimental results}

To better understand the behavior of order preserving and non-order preserving losses, we conduct two sets of experiments. The first verifies the order preservation properties that were demonstrated in section \ref{sec:order}. The second investigates the practical implications of these theoretical results.

\subsection{Simulating order preservation}

We explore the order preservation properties of commonly used loss
functions, as well as the symmetrized losses introduced in section
\ref{sec:nonconvex}. We simulate the behavior of 100 scoring functions of
varying quality over 1000 draws of 100 queries and compare the noisy risk to the
clean risk. The scoring functions are created by adding an increasing amount of noise to a perfect scorer. Further, the predictions of half of the scorers constructed in this way are scaled by $10$, to produce the effect of the counterexamples presented in proposition \ref{prop:not-ord}.
The prevalence across simulated queries varies over a range from 0.1 to 0.9 while the noise level is kept constant, $1-\gamma = 0.1$. We plot the noisy risk vs the clean risk for each loss function in figure \ref{fig:conv}.

We expect the plots corresponding to order preserving losses to have
all their points lying along a line with positive slope, showing that
the noisy risk is an affine transform of the clean loss. Figure
\ref{fig:conv} supports the findings in this work, showing that AUC is
order preserving. Interestingly, MAP does not appear to be exactly order
preserving, although the plotted points show a high rank correlation. In contrast, as demonstrated by proposition
\ref{prop:not-ord}, the logistic and exponential losses are not order
preserving; note that not all points lie along the same line.

\begin{figure*}[h]
     \centering
     \begin{subfigure}[b]{0.3\textwidth}
         \centering
         \includegraphics[width=\textwidth]{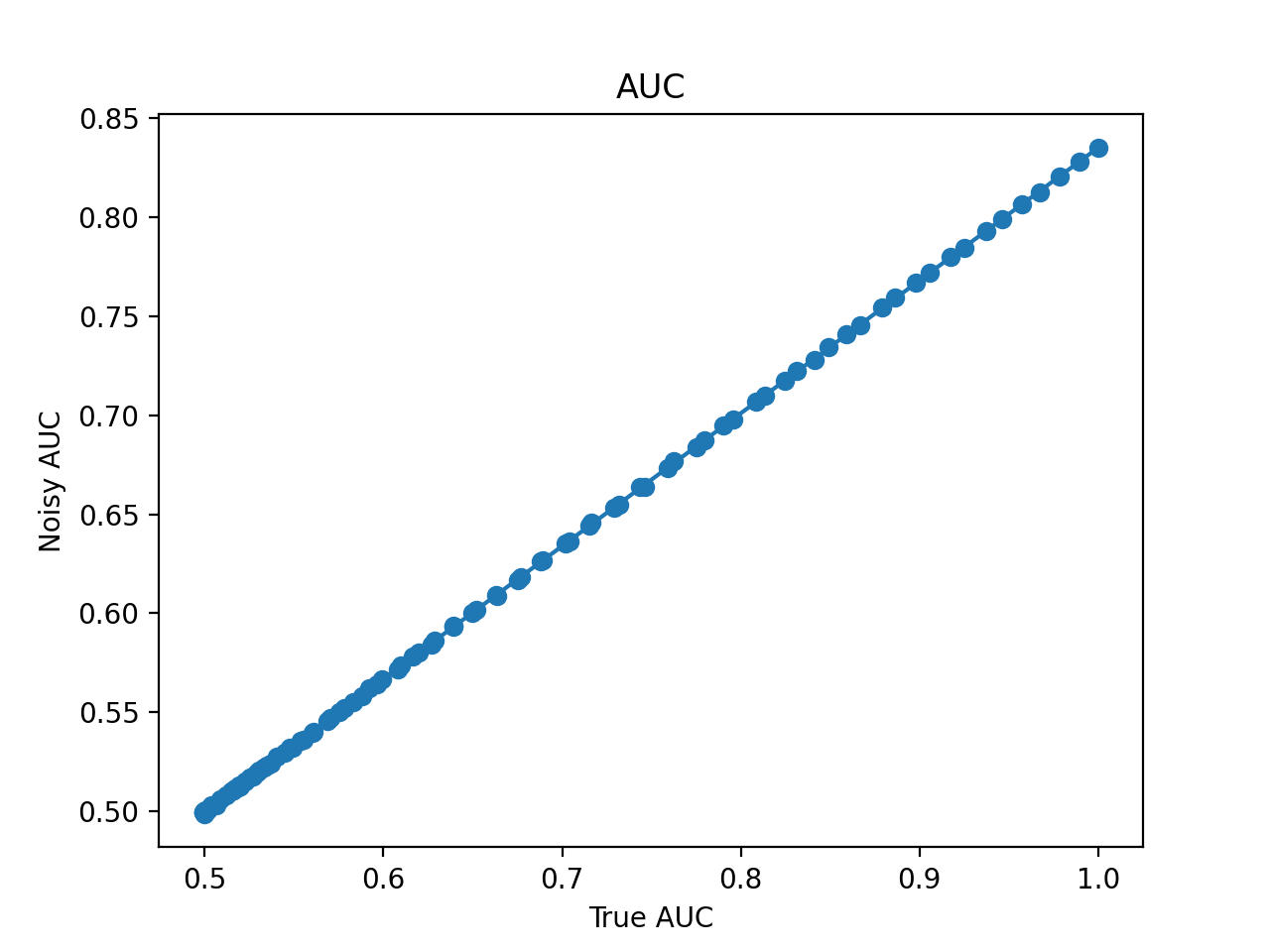}
         \caption{}
         \label{fig:auc}
     \end{subfigure}
     \hfill
     \begin{subfigure}[b]{0.3\textwidth}
         \centering
         \includegraphics[width=\textwidth]{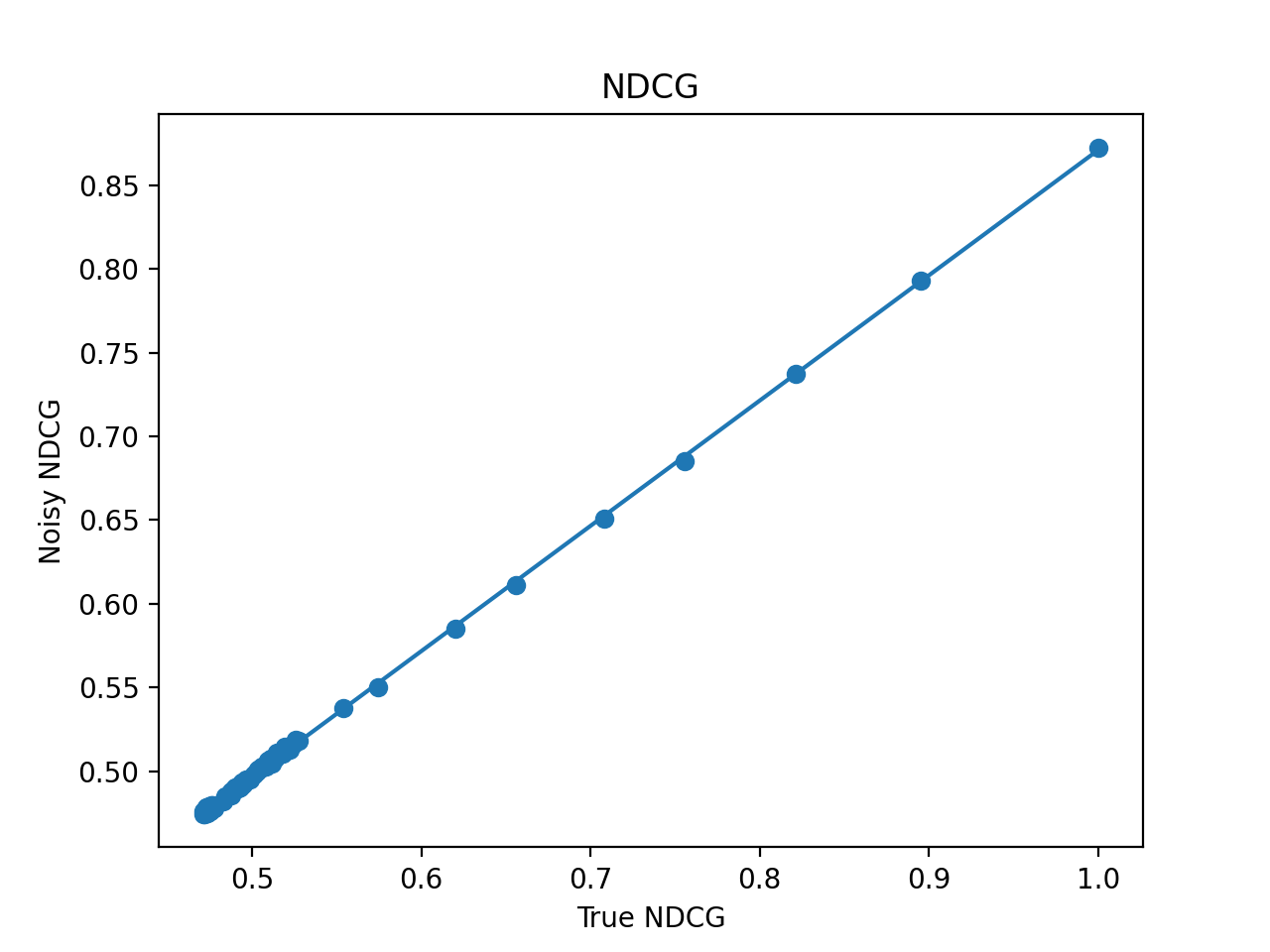}
         \caption{}
         \label{fig:ndcg}
     \end{subfigure}
     \hfill
     \begin{subfigure}[b]{0.3\textwidth}
         \centering
         \includegraphics[width=\textwidth]{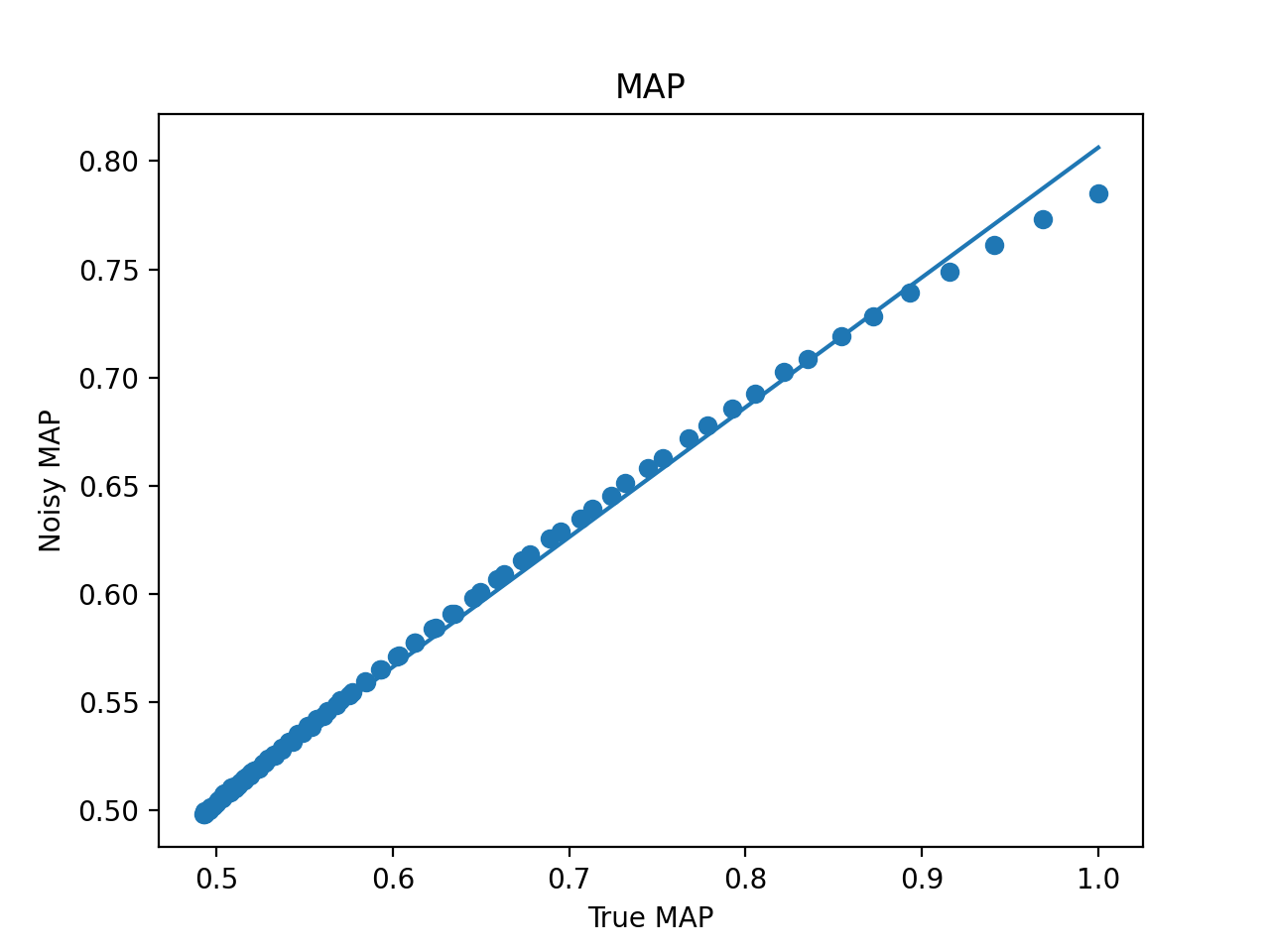}
         \caption{}
         \label{fig:map}
     \end{subfigure}
           \hfill
     \begin{subfigure}[b]{0.3\textwidth}
         \includegraphics[width=\textwidth]{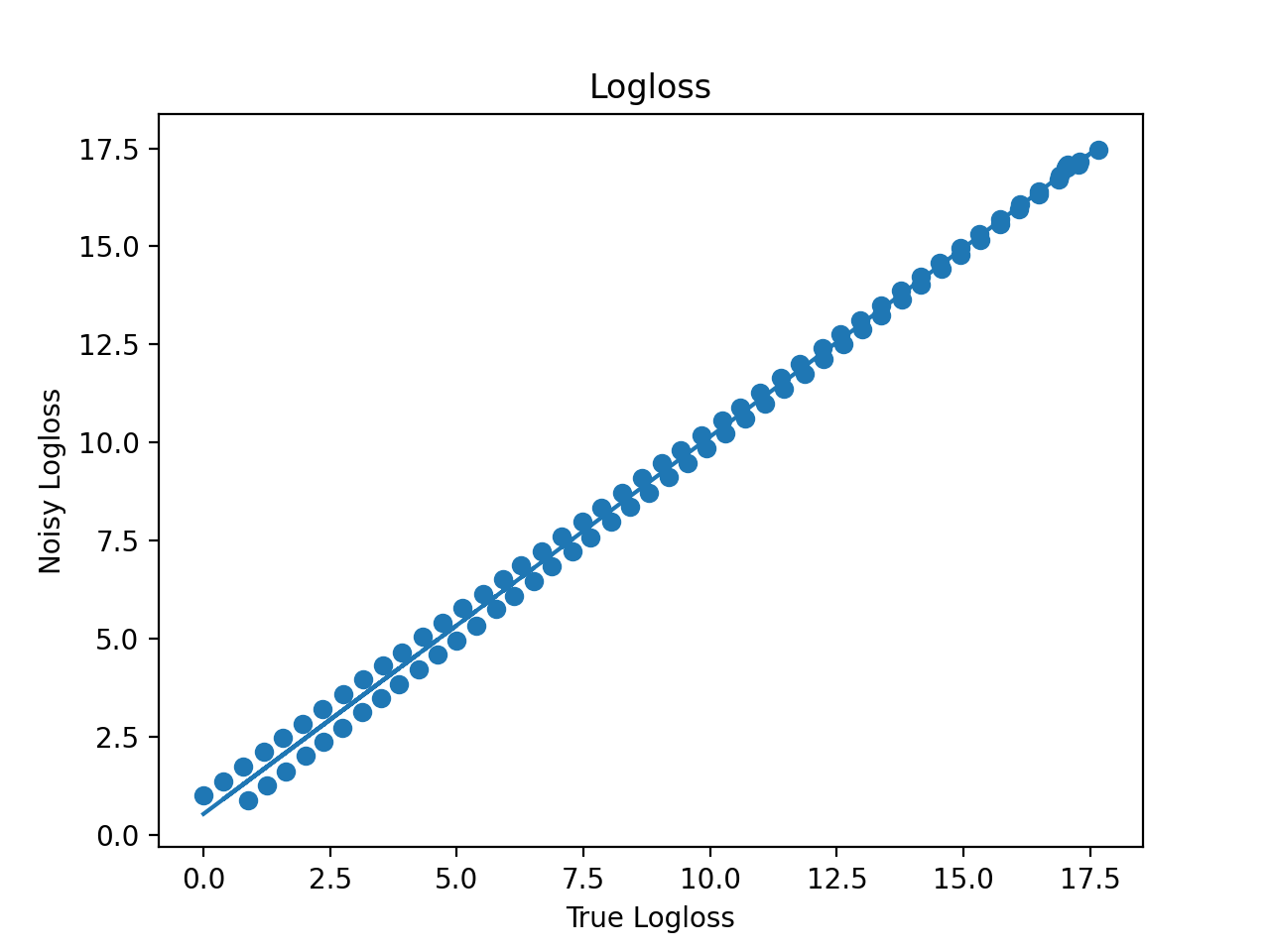}
         \caption{}
         \label{fig:logistic}
     \end{subfigure}
     \begin{subfigure}[b]{0.3\textwidth}
         \includegraphics[width=\textwidth]{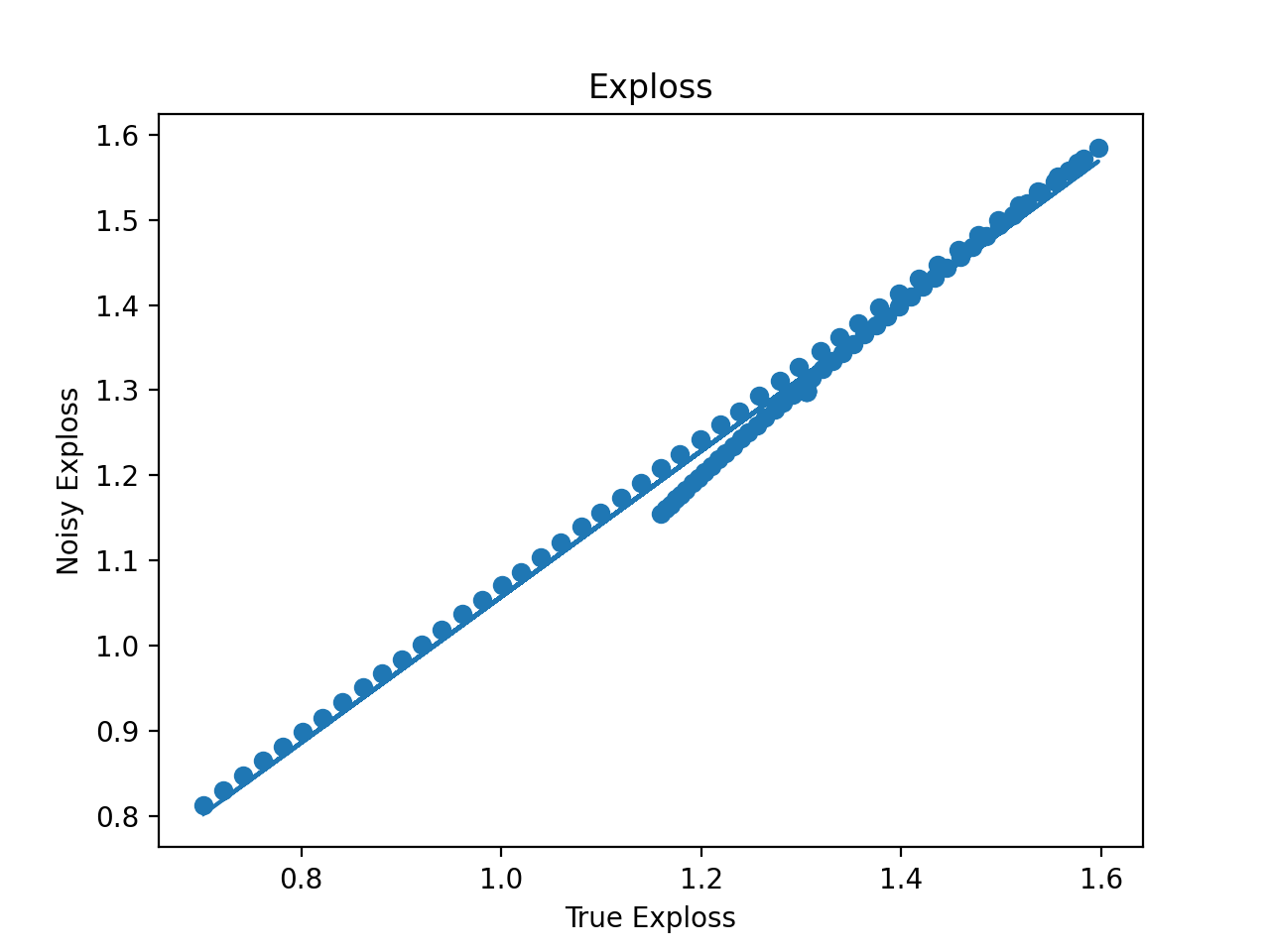}
         \caption{}
         \label{fig:exp}
     \end{subfigure}
        \caption{Simulation of the order preservation properties of
          various loss functions. The x-axis represents the clean risk, while the y-axis represents the noisy risk. For visual purposes, the line of best fit is overlayed. The exponential loss and the logistic loss are not order preserving, which can be seen in subfigures \ref{fig:logistic} and \ref{fig:exp}; notice that not all the points lie along the same line. The relationship between the noisy and true risk for MAP does not appear to be linear.}
        \label{fig:conv}
\end{figure*}

\subsection{Performance of ERM for label-symmetric losses}

In practice, we are primarily interested in selecting the optimal
scorer, rather than ordering the full collection of scorers. Motivated
by the results in section \ref{sec:nonconvex}, we explore the
practicality of the non-convex symmetrized losses introduced in
section \ref{sec:nonconvex}. We compare the performance of linear
models trained using various learning-to-rank loss functions on four datasets.

The first dataset is a synthetic learning-to-rank dataset inspired by
the experiments of \cite{mei2016}. The second is the simple 20-Newsgroups
dataset which is a collection of message board posts where the
originating group is treated as the query. The last two are the
popular LETOR datasets, MQ2007 and MQ2008, freely available from Microsoft\footnote{https://www.microsoft.com/en-us/research/project/letor-learning-rank-information-retrieval/}.

We inject class-conditional label noise over a range of $\gamma$ values for each dataset. $\gamma$ ranges over $[1, 0.9, 0.8, 0.7, 0.6, 0.51]$. We avoid $\gamma \leq 0.5$ since we do not have any theoretical justification for the selection of one loss function over another in this regime.

The synthetic datasets are generated by first sampling a $\theta_q$ at
random from an isotropic gaussian in $d$ dimensions for each
query. The document features are also sampled from an isotropic
gaussian in $d$ dimensions and the labels for each query are generated
by sampling from $\mathbb P(Y=1 \mid X) = \sigma(\langle\theta_q,\,
X\rangle)$ where $\sigma$ is the logistic function. The theoretical
results in \cite{mei2016} justify empirical risk minimization of certain non-convex losses in this context (see their theorem 4). In our experiments, the number of samples is $500$ with $5$ features.

The 20-Newsgroups dataset is preprocessed by considering a bag of words representation of the documents, removing all headers and footers to avoid leaking relevance labels into the content of the document.

We use the query-normalized versions of the MQ2007 and MQ2008 datasets and limit the labels to binary relevance. The features are further normalized to allow for faster convergence.

We train a linear model for each dataset and each choice of
$\gamma$ using four different loss functions, the commonly used
logistic and ranknet losses, in addition to their symmetrized analogs
introduced in section \ref{sec:nonconvex}. We train each model until convergence
(determined based on the estimated loss over a holdout set of the
noisy data) using the Adam optimizer \cite{kingma2014}. The learning rate and weight
decay are selected by running a grid search, with the learning rate
selected from $[\text{1e-1, 1e-2, 1e-3, 1e-4, 1e-5}]$ and the weight
decay selected from $[\text{1e-5, 1e-4, 1e-3}]$. The performance of the models over the varying noise levels are shown in figure \ref{fig:erm}.\footnote{Other optimizers led to similar results, but converged slower.}

As suggested by our theoretical results, models trained with the
symmetrized ranknet loss tend to perform better than the other models
at higher noise rates. Models trained with the more traditional losses
are competitive at the lower noise rates. For the MQ2007 and MQ2008
data sets, models trained with the symmetrized logistic loss perform
poorly (substantially worse than models trained with the other
methods) across all noise levels. The distributional assumptions of Mei et
al.'s theorem 2 and theorem 4 may not be satisfied in this case or may
be satisfied only with unfavorable constants, leading to vacuous
bounds. In particular, the gradient of the loss may not in fact be
subgaussian. Another possiblity is that the increased difficulty of the
optimization problem is not made up for by the benefit from noise tolerance for this dataset.

Despite certain loss functions not being exactly order preserving, the performance of the trained models in this section suggest that they may be approximately order preserving, in some sense. If the logistic and ranknet losses did not at all preserve the order of scoring functions in the context of class-conditional label noise, then we would see much worse performance than we see in this section for models trained using these objective functions.

\begin{figure*}[h]
     \centering
     \begin{subfigure}[b]{0.45\textwidth}
         \centering
         \includegraphics[width=\textwidth]{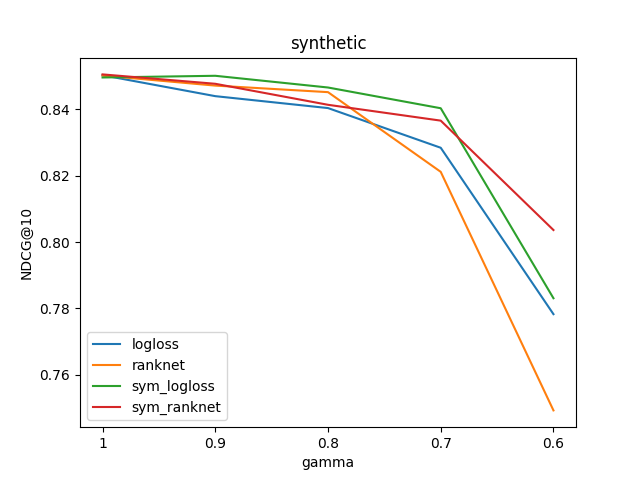}
         \caption{Synthetic dataset}
         \label{fig:synthetic}
     \end{subfigure}
     \begin{subfigure}[b]{0.45\textwidth}
         \centering
         \includegraphics[width=\textwidth]{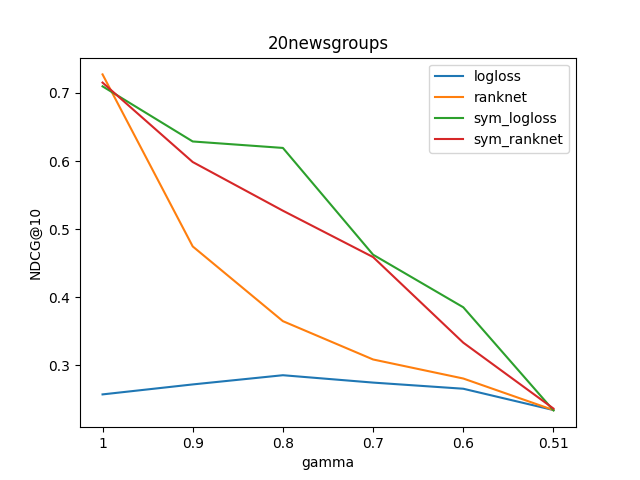}
         \caption{20-Newsgroups dataset}
         \label{fig:20newsgroups}
     \end{subfigure}
          \hfill

     \begin{subfigure}[b]{0.45\textwidth}
         \includegraphics[width=\textwidth]{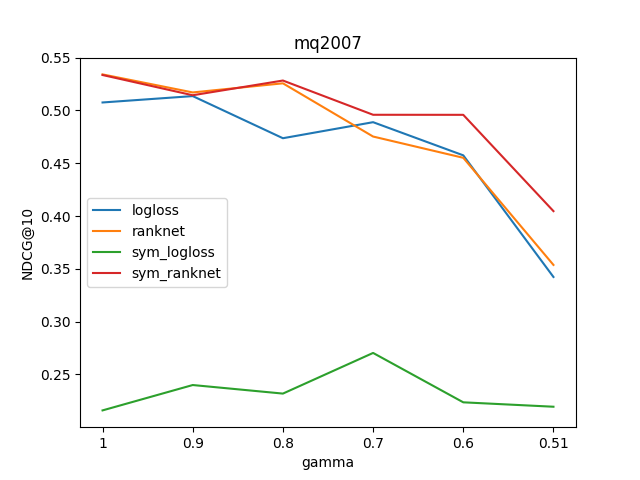}
         \caption{MQ2007 LETOR dataset}
         \label{fig:mq2007}
     \end{subfigure}
     \begin{subfigure}[b]{0.45\textwidth}
         \includegraphics[width=\textwidth]{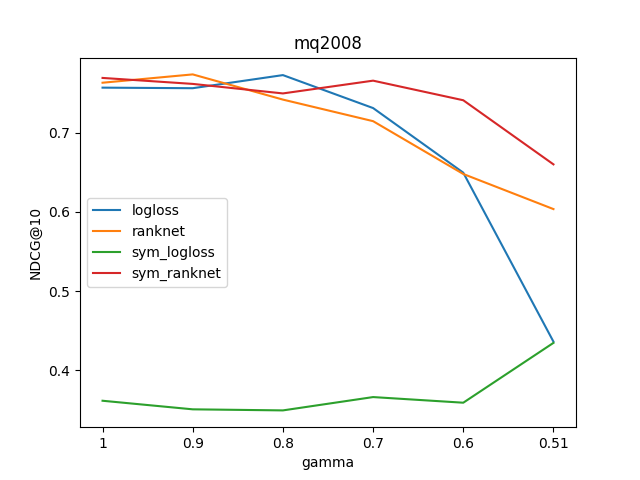}
         \caption{MQ2008 LETOR dataset}
         \label{fig:mq2008}
     \end{subfigure}
          \hfill
          \caption{NDCG@10 of models trained using various loss
            functions across four datasets and a range of noise
            conditions. $1-\gamma$ is the label flip probability, so
            as $\gamma$ decreases, the noise level increases. The
            symmetrized ranknet tends to perform best at the higher noise
            levels, while the more traditional loss functions are
            competitive at the lower noise levels. The symmetrized
            logistic loss performs poorly on the MQ2007 and MQ2008
            datasets, likely due to optimization
            difficulties or misspecified assumptions from \cite{mei2016}. MAP
            and AUC had similar behavior to NDCG@10.}
        \label{fig:erm}
      \end{figure*}

\section{Conclusion}
In this work, we showed that certain loss functions commonly used in
learning to rank applications are noise tolerant in the sense that
they are \textit{order-preserving} in the context of class-conditional
label corruptions. We further developed a sufficient
condition for a pairwise or pointwise loss to be considered
order-preserving and introduce order-preserving analogs of commonly
used loss functions. In addition, we show that empirical risk
minimization of these losses enjoys the
same properties as in the classical no-noise case; in particular, ERM
in the context of class-conditional noise is consistent and satisfies the same convergence rate
as the no-noise case, but scaled by a constant related to the
magnitude of the noise.

Experimental results support our theoretical
findings while also suggesting that some metrics that are more difficult to
analyze theoretically, NDCG and MAP, may
be order preserving.

We leave it to future work to more closely analyze additional learning to rank
losses such at MAP, NDCG and ERR. One possible route is to apply the
analysis techniques of Wang et al. \cite{wang13}  to
investigate the noise tolerance properties of NDCG. In particular,
they introduce a \textit{pseudo-expectation} of NDCG which is close to
the empirical NDCG with high probability, but is easier to analyze. A limitation of this approach is that it cannot be used to analyze NDCG losses with a finite and fixed cutoff point $k$.

Our analysis only considered the case where the noise level $\gamma$ is constant across queries. In some applications, this assumption is not likely to hold. It might be possible to show that similar results to the ones shown here hold when we instead assume that ratio between the noise level and the prevalence is constant across queries. We also leave it to future work to analyze the behavior of
learning-to-rank losses in the case of multiple relevance labels.

\bibliographystyle{ACM-Reference-Format}
\bibliography{ictir}


\begin{thebibliography}{18}


\ifx \showCODEN    \undefined \def \showCODEN     #1{\unskip}     \fi
\ifx \showDOI      \undefined \def \showDOI       #1{#1}\fi
\ifx \showISBNx    \undefined \def \showISBNx     #1{\unskip}     \fi
\ifx \showISBNxiii \undefined \def \showISBNxiii  #1{\unskip}     \fi
\ifx \showISSN     \undefined \def \showISSN      #1{\unskip}     \fi
\ifx \showLCCN     \undefined \def \showLCCN      #1{\unskip}     \fi
\ifx \shownote     \undefined \def \shownote      #1{#1}          \fi
\ifx \showarticletitle \undefined \def \showarticletitle #1{#1}   \fi
\ifx \showURL      \undefined \def \showURL       {\relax}        \fi
\providecommand\bibfield[2]{#2}
\providecommand\bibinfo[2]{#2}
\providecommand\natexlab[1]{#1}
\providecommand\showeprint[2][]{arXiv:#2}

\bibitem[\protect\citeauthoryear{Agarwal and Roth}{Agarwal and Roth}{2005}]%
        {agarwal2005}
\bibfield{author}{\bibinfo{person}{Shivani Agarwal} {and} \bibinfo{person}{Dan
  Roth}.} \bibinfo{year}{2005}\natexlab{}.
\newblock \showarticletitle{{Learnability of Bipartite Ranking Functions}}.
\newblock \bibinfo{journal}{\emph{Lecture Notes in Computer Science}}
  (\bibinfo{year}{2005}), \bibinfo{pages}{16--31}.
\newblock
\showISSN{0302-9743}
\urldef\tempurl%
\url{https://doi.org/10.1007/11503415\_2}
\showDOI{\tempurl}


\bibitem[\protect\citeauthoryear{Angelopoulos, Bates, Candès, Jordan, and
  Lei}{Angelopoulos et~al\mbox{.}}{2021}]%
        {angelopoulos2021}
\bibfield{author}{\bibinfo{person}{Anastasios~N Angelopoulos},
  \bibinfo{person}{Stephen Bates}, \bibinfo{person}{Emmanuel~J Candès},
  \bibinfo{person}{Michael~I Jordan}, {and} \bibinfo{person}{Lihua Lei}.}
  \bibinfo{year}{2021}\natexlab{}.
\newblock \showarticletitle{{Learn then Test: Calibrating Predictive Algorithms
  to Achieve Risk Control}}.
\newblock \bibinfo{journal}{\emph{arXiv}} (\bibinfo{year}{2021}).
\newblock
\showeprint{2110.01052}


\bibitem[\protect\citeauthoryear{Bates, Angelopoulos, Lei, Malik, and
  Jordan}{Bates et~al\mbox{.}}{2021}]%
        {bates2021}
\bibfield{author}{\bibinfo{person}{Stephen Bates}, \bibinfo{person}{Anastasios
  Angelopoulos}, \bibinfo{person}{Lihua Lei}, \bibinfo{person}{Jitendra Malik},
  {and} \bibinfo{person}{Michael~I Jordan}.} \bibinfo{year}{2021}\natexlab{}.
\newblock \showarticletitle{{Distribution-Free, Risk-Controlling Prediction
  Sets}}.
\newblock \bibinfo{journal}{\emph{arXiv}} (\bibinfo{year}{2021}).
\newblock
\showeprint{2101.02703}


\bibitem[\protect\citeauthoryear{Burges}{Burges}{2010}]%
        {burges2010}
\bibfield{author}{\bibinfo{person}{CJC Burges}.}
  \bibinfo{year}{2010}\natexlab{}.
\newblock \showarticletitle{{From RankNet to LambdaRank to LambdaMART: An
  Overview}}.
\newblock  (\bibinfo{year}{2010}).
\newblock


\bibitem[\protect\citeauthoryear{Burges, Ragno, and Le}{Burges
  et~al\mbox{.}}{2006}]%
        {burges2006}
\bibfield{author}{\bibinfo{person}{Christopher Burges}, \bibinfo{person}{Robert
  Ragno}, {and} \bibinfo{person}{Quoc~V Le}.} \bibinfo{year}{2006}\natexlab{}.
\newblock \showarticletitle{{Learning to Rank with Nonsmooth Cost Functions}}.
\newblock \bibinfo{journal}{\emph{Advances in Neural Information Processing
  Systems}} (\bibinfo{year}{2006}).
\newblock


\bibitem[\protect\citeauthoryear{Dehghani, Zamani, Severyn, Kamps, and
  Croft}{Dehghani et~al\mbox{.}}{2017}]%
        {dehghani2017}
\bibfield{author}{\bibinfo{person}{Mostafa Dehghani}, \bibinfo{person}{Hamed
  Zamani}, \bibinfo{person}{Aliaksei Severyn}, \bibinfo{person}{Jaap Kamps},
  {and} \bibinfo{person}{W~Bruce Croft}.} \bibinfo{year}{2017}\natexlab{}.
\newblock \showarticletitle{{Neural Ranking Models with Weak Supervision}}
  \emph{(\bibinfo{series}{the 40th International ACM SIGIR Conference})}.
  \bibinfo{pages}{65 -- 74}.
\newblock
\showISBNx{9781450350228}
\urldef\tempurl%
\url{https://doi.org/10.1145/3077136.3080832}
\showDOI{\tempurl}


\bibitem[\protect\citeauthoryear{Devroye, Gyorfi, and Lugosi}{Devroye
  et~al\mbox{.}}{1997}]%
        {devroye1997}
\bibfield{author}{\bibinfo{person}{Luc Devroye}, \bibinfo{person}{Laszlo
  Gyorfi}, {and} \bibinfo{person}{Gabor Lugosi}.}
  \bibinfo{year}{1997}\natexlab{}.
\newblock \showarticletitle{{A Probabilistic Theory of Pattern Recognition}}.
\newblock \bibinfo{journal}{\emph{Discrete Applied Mathematics}}
  \bibinfo{volume}{73}, \bibinfo{number}{2} (\bibinfo{year}{1997}),
  \bibinfo{pages}{192--194}.
\newblock
\showISSN{0166-218X}
\urldef\tempurl%
\url{https://doi.org/10.1016/s0166-218x(97)81417-5}
\showDOI{\tempurl}


\bibitem[\protect\citeauthoryear{Gneiting, Balabdaoui, and Raftery}{Gneiting
  et~al\mbox{.}}{2007}]%
        {gneiting}
\bibfield{author}{\bibinfo{person}{Tilmann Gneiting}, \bibinfo{person}{Fadoua
  Balabdaoui}, {and} \bibinfo{person}{Adrian~E. Raftery}.}
  \bibinfo{year}{2007}\natexlab{}.
\newblock \showarticletitle{{Probabilistic forecasts, calibration and
  sharpness}}.
\newblock \bibinfo{journal}{\emph{Journal of the Royal Statistical Society:
  Series B (Statistical Methodology)}} \bibinfo{volume}{69},
  \bibinfo{number}{2} (\bibinfo{year}{2007}), \bibinfo{pages}{243--268}.
\newblock
\showISSN{1467-9868}
\urldef\tempurl%
\url{https://doi.org/10.1111/j.1467-9868.2007.00587.x}
\showDOI{\tempurl}


\bibitem[\protect\citeauthoryear{Györfi, Kohler, Krzyżak, and Walk}{Györfi
  et~al\mbox{.}}{2002}]%
        {gyorfi2002}
\bibfield{author}{\bibinfo{person}{László Györfi}, \bibinfo{person}{Michael
  Kohler}, \bibinfo{person}{Adam Krzyżak}, {and} \bibinfo{person}{Harro
  Walk}.} \bibinfo{year}{2002}\natexlab{}.
\newblock \showarticletitle{{A Distribution-Free Theory of Nonparametric
  Regression}}.
\newblock  (\bibinfo{year}{2002}).
\newblock
\showISSN{0172-7397}
\urldef\tempurl%
\url{https://doi.org/10.1007/b97848}
\showDOI{\tempurl}


\bibitem[\protect\citeauthoryear{Haddad and Ghosh}{Haddad and Ghosh}{2019}]%
        {haddad2019}
\bibfield{author}{\bibinfo{person}{Dany Haddad} {and} \bibinfo{person}{Joydeep
  Ghosh}.} \bibinfo{year}{2019}\natexlab{}.
\newblock \showarticletitle{Learning More From Less: Towards Strengthening Weak
  Supervision for Ad-Hoc Retrieval}. In \bibinfo{booktitle}{\emph{Proceedings
  of the 42nd International ACM SIGIR Conference on Research and Development in
  Information Retrieval}} (Paris, France) \emph{(\bibinfo{series}{SIGIR'19})}.
  \bibinfo{publisher}{Association for Computing Machinery},
  \bibinfo{address}{New York, NY, USA}, \bibinfo{pages}{857–860}.
\newblock
\showISBNx{9781450361729}
\urldef\tempurl%
\url{https://doi.org/10.1145/3331184.3331272}
\showDOI{\tempurl}


\bibitem[\protect\citeauthoryear{Kingma and Ba}{Kingma and Ba}{2014}]%
        {kingma2014}
\bibfield{author}{\bibinfo{person}{Diederik~P Kingma} {and}
  \bibinfo{person}{Jimmy Ba}.} \bibinfo{year}{2014}\natexlab{}.
\newblock \showarticletitle{{Adam: A Method for Stochastic Optimization}}.
\newblock \bibinfo{journal}{\emph{arXiv}} (\bibinfo{year}{2014}).
\newblock
\showeprint{1412.6980}


\bibitem[\protect\citeauthoryear{Mei, Bai, and Montanari}{Mei
  et~al\mbox{.}}{2016}]%
        {mei2016}
\bibfield{author}{\bibinfo{person}{Song Mei}, \bibinfo{person}{Yu Bai}, {and}
  \bibinfo{person}{Andrea Montanari}.} \bibinfo{year}{2016}\natexlab{}.
\newblock \showarticletitle{{The Landscape of Empirical Risk for Non-convex
  Losses}}.
\newblock \bibinfo{journal}{\emph{arXiv}} (\bibinfo{year}{2016}).
\newblock
\showeprint{1607.06534}


\bibitem[\protect\citeauthoryear{Natarajan, Dhillon, Ravikumar, and
  Tewari}{Natarajan et~al\mbox{.}}{2018}]%
        {natarajan2018}
\bibfield{author}{\bibinfo{person}{Nagarajan Natarajan},
  \bibinfo{person}{Inderjit~S. Dhillon}, \bibinfo{person}{Pradeep Ravikumar},
  {and} \bibinfo{person}{Ambuj Tewari}.} \bibinfo{year}{2018}\natexlab{}.
\newblock \showarticletitle{{Cost-Sensitive Learning with Noisy Labels}}.
\newblock \bibinfo{journal}{\emph{Journal of Machine Learning Research}}
  (\bibinfo{year}{2018}).
\newblock


\bibitem[\protect\citeauthoryear{Wainwright}{Wainwright}{2019}]%
        {wainright2019}
\bibfield{author}{\bibinfo{person}{Martin~J Wainwright}.}
  \bibinfo{year}{2019}\natexlab{}.
\newblock \showarticletitle{{High-Dimensional Statistics}}.
\newblock  (\bibinfo{date}{12} \bibinfo{year}{2019}), \bibinfo{pages}{1--20}.
\newblock
\urldef\tempurl%
\url{https://doi.org/10.1017/9781108627771.001}
\showDOI{\tempurl}


\bibitem[\protect\citeauthoryear{Wang, Wang, Li, He, Liu, and Chen}{Wang
  et~al\mbox{.}}{2013}]%
        {wang13}
\bibfield{author}{\bibinfo{person}{Yining Wang}, \bibinfo{person}{Liwei Wang},
  \bibinfo{person}{Yuanzhi Li}, \bibinfo{person}{Di He},
  \bibinfo{person}{Tie-Yan Liu}, {and} \bibinfo{person}{Wei Chen}.}
  \bibinfo{year}{2013}\natexlab{}.
\newblock \showarticletitle{{A Theoretical Analysis of NDCG Type Ranking
  Measures}}.
\newblock \bibinfo{journal}{\emph{arXiv}} (\bibinfo{year}{2013}).
\newblock
\showeprint{1304.6480}


\bibitem[\protect\citeauthoryear{Zamani and Croft}{Zamani and Croft}{2018}]%
        {zamani2018ictir}
\bibfield{author}{\bibinfo{person}{Hamed Zamani} {and} \bibinfo{person}{W~Bruce
  Croft}.} \bibinfo{year}{2018}\natexlab{}.
\newblock \bibinfo{booktitle}{\emph{{On the Theory of Weak Supervision for
  Information Retrieval}}}.
\newblock \bibinfo{publisher}{ACM}.
\newblock
\showISBNx{978-1-4503-5656-5}
\urldef\tempurl%
\url{https://doi.org/10.1145/3234944.3234968}
\showDOI{\tempurl}


\bibitem[\protect\citeauthoryear{Zamani, Dehghani, Croft, Learned-Miller, and
  Kamps}{Zamani et~al\mbox{.}}{2018}]%
        {zamani2018acm}
\bibfield{author}{\bibinfo{person}{Hamed Zamani}, \bibinfo{person}{Mostafa
  Dehghani}, \bibinfo{person}{W~Bruce Croft}, \bibinfo{person}{Erik
  Learned-Miller}, {and} \bibinfo{person}{Jaap Kamps}.}
  \bibinfo{year}{2018}\natexlab{}.
\newblock \showarticletitle{{From Neural Re-Ranking to Neural Ranking}}
  \emph{(\bibinfo{series}{the 27th ACM International Conference})}.
  \bibinfo{pages}{497 -- 506}.
\newblock
\showISBNx{9781450360142}
\urldef\tempurl%
\url{https://doi.org/10.1145/3269206.3271800}
\showDOI{\tempurl}


\bibitem[\protect\citeauthoryear{Zhang, Lee, and Agarwal}{Zhang
  et~al\mbox{.}}{2021}]%
        {zhang2021}
\bibfield{author}{\bibinfo{person}{Mingyuan Zhang}, \bibinfo{person}{Jane Lee},
  {and} \bibinfo{person}{Shivani Agarwal}.} \bibinfo{year}{2021}\natexlab{}.
\newblock \showarticletitle{{Learning from Noisy Labels with No Change to the
  Training Process}}.
\newblock \bibinfo{journal}{\emph{Proceedings of the 38th International
  Conference on Machine Learning}} (\bibinfo{year}{2021}).
\newblock
\urldef\tempurl%
\url{https://proceedings.mlr.press/v139/zhang21k.html}
\showURL{%
\tempurl}


\end{thebibliography}


\end{document}